\documentclass[12pt]{article}

\usepackage[utf8]{inputenc}

\usepackage{hyperref}

\usepackage{amsmath}
\usepackage{amssymb}
\usepackage{xspace}
\usepackage{graphicx, subfig}               
\usepackage{setspace}
\usepackage[round]{natbib}

\newtheorem{theorem}{Theorem}[section]
\newtheorem{proof}{Proof}[section]
\newtheorem{corollary}[theorem]{Corollary}

\usepackage{color}

\newcommand{\cutac}[1]{\ac [cut] \ca} 

\newcommand{\ma}{\color{magenta} }  
\newcommand{\am}{ \color{black}} 
\newcommand{\cutma}[1]{\ma [cut] \am} 

\newcommand{\bA}{\mathbf{A}}
\newcommand{\bD}{\mathbf{D}}

\newcommand{\R}{\mathbb{R}}
\newcommand{\bY}{\mathbf{Y}}

\newcommand{\bb}{\mathbf{b}}
\newcommand{\bc}{\mathbf{c}}
\newcommand{\vf}{\mathbf{f}}  

\newcommand{\bu}{\mathbf{u}}

\newcommand{\by}{\mathbf{y}}

\footnotetext[1]{Centro de Investigaci\'on en Matem\'aticas (CIMAT), Jalisco S/N, Valenciana, Guanajuato, GT, 36023, MEXICO.
\textit{jac@cimat.mx, marcos@cimat.mx, moreles@cimat.mx}}
\footnotetext[2]{Corresponding author.}

\begin{document}

\title{Numerical posterior distribution error control\\
and expected Bayes Factors in the Bayesian\\
Uncertainty Quantification of inverse problems}

\author{J. Andr\'es Christen$^{1,2}$, Marcos A. Capistr\'an$^{1}$\\
and Miguel \'Angel Moreles$^{1}$}


\date{29 AUG 2017}

\maketitle

\bibliographystyle{chicago}

\begin{abstract}
In bayesian UQ most relevant cases of forward maps
(FM, or regressor function) are defined in terms of a system of (O, P)DE's with intractable
solutions.  These necessarily involve a numerical method to find approximate versions of such
solutions which lead to a numerical/approximate posterior distribution.  In the past decade, several
results have been published on the regularity conditions required
to ensure converge of the numerical to the theoretical posterior.
However, more practical guidelines are needed to ensure a suitable working numerical posterior.
\cite{Capistran2016} prove for ODEs that the Bayes Factor (BF) of the approximate vs the theoretical model tends
to 1 in the same order as the numerical method approximation order.  In this work we generalize the latter
paper in that we consider 1) the use of expected BFs, 2) also PDEs, 3) correlated observations,
which results in, 4) more practical and workable guidelines in a more realistic multidimensional setting.
The main result is a bound on
the absolute global errors to be tolerated by the FM numerical solver, which we illustrate
with some examples.  Since the BF is kept near 1 we expect that
the resulting numerical posterior is basically  indistinguishable from the theoretical
posterior, even though we are using an approximate numerical FM.  The method is illustrated
with an ODE and a PDE example, using synthetic data.
\end{abstract}



\newpage

\section{Introduction}\label{sec:intro}

Bayesian Uncertainty Quantification (UQ) has attracted substantial attention in recent years,
covering a wide range of applications both in well established fields as well as in emerging
areas.  Some recent examples may be found in
\cite{ZHUetAl2011, CAIetAl2011, FALLetAl2011, CHAMAetAl2012,
NISSINENetAl2011, KOZAWAetAl2012,
CUIetAl2011, WAN&ZABARAS2011, HAZELTON2010, KAIPIO&FOX2011}.  For reviews on the subject see  
\cite{KAIPIO&FOX2011, WATZENIG&FOX2009, WOODBURY2011,Fox2013}.  

The usual parametric (finite dimensional)
Bayesian formulation of Inverse Problems, in broad terms, is that 
given a Forward Map (FM) $F_\theta$, a noise model is assumed for the observations
$y_j \sim G_{\sigma}(F^j_\theta)$, for some noise level $\sigma$
and typically additive gaussian errors with known standard deviation $\sigma$  are assumed.
This observation model creates a probability density of all data $\bY$ given all parameters
$\Phi$ namely $P_{ \bY | \Phi} (\by | \theta,\sigma)$.  For fixed data $\by$ this forms
the likelihood, regarding the latter as a function of $\theta,\sigma$, and is the
basis of the statistical analysis of Inverse Problems.  Using Bayesian inference
one establishes a prior distribution $P_{\Phi}(\theta,\sigma)$ and defines the posterior
distribution
\begin{equation}\label{eqn.exact_post}
P_{ \Phi | \bY }( \theta,\sigma | \by ) =  \frac{P_{ \bY | \Phi }( \by |  \theta,\sigma) P_{\Phi}(\theta,\sigma)}{P_{\bY} (\by)} .
\end{equation}
This probability distribution on the unknowns $\theta$ and $\sigma$ quantifies the uncertainty
on the possible values for these parameters coherent with the data $\by$.
However, the common denominator in this particular Bayesian inference problem is
that we do not have an analytical or computationally simple and precise implementation
of the FM.

Instead, a numerical approach is required to create a solver
and find a numerical approximation of the FM $F^{\alpha}_\theta$, for some discretization
parameter $\alpha$ (eg. step size, grid norm, terms in a series, etc. concrete examples will be
given in section  \ref{sec.setting}).  Since we can only use the numeric approximation,
this in turn leads to a numeric likelihood $P^\alpha_{ \bY | \Phi} (\by | \theta,\sigma)$
and therefore a numeric posterior
\begin{equation}\label{eqn.num_post}
P^\alpha_{ \Phi | \bY }( \theta,\sigma | \by ) =  \frac{P^\alpha_{ \bY | \Phi }( \by |  \theta,\sigma) P_{\Phi}(\theta,\sigma)}{P^\alpha_{\bY} (\by)} .
\end{equation}
$P^\alpha_{\bY} (\by) = \int P^\alpha_{ \bY | \Phi }( \by |  \theta,\sigma) P_{\Phi}(\theta,\sigma) d\theta d\sigma$
and 
$P_{\bY} (\by) = \int P_{ \bY | \Phi }( \by |  \theta,\sigma) P_{\Phi}(\theta,\sigma) d\theta d\sigma$ are the normalization constants of the two models,  also called the \emph{marginal likelihoods} of data $\by$. 

Numerical methods are designed so as, if the discretization tends to zero $|\alpha| \rightarrow 0$,
for some norm or functional $|\cdot|$, then the numeric FM tends to the theoretical FM at some
order $O(|\alpha|^p)$.  This is the global error control for the numerical method or solver, which we dicuss in detail in section \ref{sec.global_err}.  However,
it is of great interest to prove that the same happens with the theoretical
vs the numeric posteriors in (\ref{eqn.exact_post}) and (\ref{eqn.num_post}) respectively, in order
to make sense of our Bayesian approach.

Recently a number of papers have dealt with this problem in a theoretical sense by establishing
regularity conditions so as 
$$
\lim_{|\alpha| \rightarrow 0}
||  P^\alpha_{ \Phi | \bY }( \theta,\sigma | \by ) - P_{ \Phi | \bY }( \theta,\sigma | \by ) || = 0 ,
$$
for some (eg. Hellinger) measure metric $|| \cdot ||$; see \cite{COTTER2010} for a review.
This forms a sound theoretical basis for the Bayesian analysis of inverse problems.  However,
in applications we have to choose a  discretization $\alpha$.  More practical guidelines
are needed to choose the numerical solver precision and how this controls the level
of approximation between  $P^\alpha_{ \Phi | \bY }( \theta,\sigma | \by ) $ and
$P_{ \Phi | \bY }( \theta,\sigma | \by ) $.

On the other hand \cite{Capistran2016} present an approach to address the above problem
using Bayes factors (BF; the odds in favor) of the numerical model vs the theoretical
model (further details will be given in section  \ref{sec.setting}).  With equal prior probability
for both models, this BF is $\frac{P^\alpha_{\bY} (\by)}{P_{\bY} (\by)}$.
In an ODE framework, these odds are proved in \cite{Capistran2016} to converge to 1 (that is, both models
would be equal) in the same order as the numerical solver used.
For high order solvers \cite{Capistran2016} illustrates, by reducing the step size
in the numerical solver,  that there should exist a point at which the BF is basically 1,
but for fixed  discretization $\alpha$ (step size) greater than zero.  This is the main point made by \cite{Capistran2016}:
it could be possible to calculate,  for solver orders of 2 or more, a threshold for the tolerance such
that the \textit{numerical} posterior is basically equal to the theoretical posterior so, although we
are using an approximate FM, the resulting posterior is error free.
\cite{Capistran2016} illustrate, with some examples, that
such optimal solver  discretization leads to basically no differences in the numerical and
the theoretical posterior (since the BF is basically 1).  Moreover, since for most solvers its
computational complexity goes to infinity as $|\alpha| \rightarrow 0$, using the optimal $\alpha$ led
in their examples to a 90\% save in CPU time.

However,  \cite{Capistran2016} still has a number of shortcomings.  First, it depends crucially
on estimating the normalizing constants $P^\alpha_{\bY} (\by)$ from Monte Carlo
samples of the unnormalized posterior, for a range of  discretizations $|\alpha|$.  This is
a very complex estimation problem and is the subject of current research and is in fact
very difficult to reliably estimate these normalizing constants in mid to high dimension
problems.  Second, \cite{Capistran2016} approach is as yet incomplete
since one would need to decrease
$|\alpha|$ systematically, calculating  $P^\alpha_{\bY} (\by)$ to eventually estimate
$P_{\bY} (\by)$, which in turn will pin point a  discretization at which both models are
indistinguishable.  Being this a second complex estimation problem, the main difficulty
here is that one has already calculated the posterior for small $|\alpha|$ and therefore
it renders useless the selection of the optimal step size.

To improve on \cite{Capistran2016}, the idea of this paper
is to consider the \textit{expected} value of the BFs, before data is observed.  We will try to bound this
expected BF to find general guidelines to establish error bounds on the numerical solver,
depending on the specific problem at hand and the sample design used, but not on particular
data.  These guidelines will be solely regarding the forward map and, although conservative,
represent useful bounds to be used in practice.

We do not discuss
in this paper the infinite dimension counterpart of this approach,
of interest when inference is needed over function spaces as it is the case in some general PDE
inverse problems, see for example \cite{COTTERetAl2009, Dunlop2015} and references therein.
As mentioned above, we restrict ourselves to the finite dimensional parametric case where we
establish our results.

The paper is organized as follows. Our formal setting will be discussed in section \ref{sec.setting}. 
In section \ref{sec:main} we present our main result, including several comments
of some implications and practical guidelines for its use.  In sections \ref{sec:exaODE} and~\ref{sec:exaPDE}
we present prove of concept examples, considering an ODE and a PDE, respectively.
In both cases, using error estimated on the numeric forward maps we were able to very sustantially
reduce CPU time while obtaining basically the same posterior.
Finally, a discussion of the paper is presented in section \ref{sec:discussion}.

\section{Setting}\label{sec.setting}

Assume that we observe a process $\by = (y_1,\dots,y_n)$ at \textit{locations}
$x_1, \dots, x_n \in D \subset \R^m$.    
This is a general setting, to include ODEs and PDEs and other inverse problems, in which the
domain may include, for example, space and time: $x_i = [ (z_{ix}, z_{iy}), t_i ]$.  That is, $x_i$
is an observation at coordinates $(z_{ix}, z_{iy})$ and at time $t_i$, etc.

We assume that the Forward Map $F_{\theta} : \R^m  \rightarrow \R^q$ is well defined for all
parameters $\theta$ where $\theta \in A \subset \R^d$.
Typically, as mentioned above, $F_{\theta}(x)$, for all $x \in D$, is the solution of a system of
ODE's or PDE's. 
This means that $F_{\theta}(x)$ are the $q$ state variables representing the solution of the
ODE or PDE system, with parameters $\theta$, at location $x$.
In many cases, specailly dealing with PDEs, the actual unknown is a function in which the inference
problem at hand is infinite dimensional.  As mentioned in the introduction, in this paper we confine
ourselves to the finite dimensional parametric problem, that is, the unknown is $\theta$ of dimension $d$.  
The initial or boundary conditions are taken as known, although these may be turned to be part of the unknown parameters, using common techniques.

Let $f: \R^q \rightarrow \R$ be the observational functional, in the sense that
$y_i$ is an observation of $f(F_{\theta}(x_i))$.  For example, $f(F_{\theta}(x_i))$ is one particular
state variable, for which we have observations. We only consider univariate observations at
each location $x_i$.

We assume gaussian errors on the observations, however, we consider the possibility of correlated observations, namely, let
$\vf_{\theta} = ( f(F_{\theta}(x_1)), \ldots , f(F_{\theta}(x_n)) )'$ then
$$
\by \mid \theta, \sigma^2, \bA \sim N_n (  \vf_{\theta}, \sigma^{-2} \bA ),
$$
where $\sigma^{-2} \bA$ is the \textit{precision} matrix (inverse of the variance-covariance matrix)
of the $n$-dimensional Gaussian distribution
with mean $\vf_{\theta}$.  $\bA^{-1}$ is a correlation matrix with some correlation structure
on $D$.  $\bA^{-1}$ will be taken as known, but $\sigma^{2}$ will be considered unknown in general.
This is a particular covariance structure, very common in time series or spatial statistics.
Indeed, if $\bA = \mathbf{I}$ we go back to the common uncorrelated error structure.
A specific example is to use a correlation function $\rho$ an a metric $d(x_i,x_j)$ in $D$ to
give $\bA^{-1} = (\rho(d(x_i,x_j)))$, and therefore correlation decreases as the location points separate.
This is an isotropic correlation structure.  Many other structures may be considered and this has been
extensively studied in the statistics literature~\cite{christakos:1992}.  Note that the marginal distribution
for the observations is
\begin{equation}\label{eqn.obsmodel}
y_i  = f(F_{\theta}(x_i)) + \varepsilon_i ,~~ \varepsilon_i \sim \mathcal{N}(0, \sigma b_i) ,
\end{equation}
where $b_i^2 =  [A^{-1}]_{ii}$, that is, $\sigma b_i$ is the standard error of the noise.

Let also
$\vf_{\theta}^{\alpha} = ( f(F_{\theta}^{\alpha}(x_1)), \ldots , f(F_{\theta}^{\alpha}(x_n)) )'$ and
$\by \mid \theta, \sigma^2, \bA \sim N_n (  \vf_{\theta}^{\alpha}, \sigma^{-2} \bA )$ for
the numerical model.  From this the likelihood of the approximated model is
\begin{equation}\label{eqn.ll}
P^\alpha_{ \bY | \Phi} (\by | \theta,\sigma) =
(2 \pi \sigma^2)^{-\frac{n}{2}} \left| \bA \right|^{ \frac{1}{2} }
\exp\left\{-\frac{1}{2 \sigma^2} (\by - \vf_{\theta}^{\alpha})' \bA (\by - \vf_{\theta}^{\alpha})\right\} ,
\end{equation}
with the corresponding expression for the exact model.  

Regarding the prior distribution for the parameters $P_{\Phi}(\theta,\sigma)$, we assume
that $P_{\Phi}(\theta,\sigma) = g(\sigma) P_\Theta (\theta)$.  A common and reasonable
assumption in which the prior on the observational noise is independent
of the prior regarding the actual model parameters. 

\subsection{Global Error Control}\label{sec.global_err}

We assume we use a numerical method to obtain an approximation of the Forward Map,
that is $F^{\alpha}_{\theta}(x_i)$, for some  discretization $\alpha$.
As a general setting for approximation strategies, the 
numerical method or solver  discretization $\alpha$ may now be multidimensional.
We assume that the global error control of the solver states that
\begin{equation}\label{eqn:global_error}
|| F^{\alpha}_{\theta}(x_i) -  F_{\theta}(x_i) || \leq K_{x_i , \theta} |\alpha|^p ,
\end{equation}
for some functional $| \cdot |$ for the  discretization.  That is, the solver is of order $p$.
For example $\alpha = ( \Delta x, \Delta t )$ in finite element PDE solvers, etc. and,
perhaps $|\alpha| = sup |\alpha_i|$.  In section \ref{sec:exaODE} we present
an ODE example in which $\alpha = \Delta t$, ie. the time step size, $| \cdot |$ is the identity
and the solver has order $p=4$.  Moreover, in section \ref{sec:exaPDE} we present
a PDE example where $\alpha = ( \Delta x, \Delta t )$, $|\alpha| = \Delta x$ and
the solver has order $p=2$.

Using this asymptotic behavior of the global error and assuming that $f$ is differentiable, we can write
\begin{equation}\label{eqn:Dh}
f(F^{\alpha}_{\theta}(x)) - f(F_{\theta}(x)) =
\nabla f \left( F_{\theta}(x) \right)(F^{\alpha}_{\theta}(x) -F_{\theta}(x)) +  O(|\alpha|^{2p}) = O(|\alpha|^p) . 
\end{equation}
for all $x \in D$ and $\theta \in A$.  We further assume that $\nabla f$ is bounded and that
$K_{x_i , \theta} < K'$ for all locations $x_i$ and all $\theta \in A$.
From this we conclude that the global error is
\begin{equation}\label{eqn.global_error}
| f(F^{\alpha}_{\theta}(x_i)) - f(F_{\theta}(x_i)) | \leq K |\alpha|^p ,
\end{equation}
for some global $K$.  This approximation is also of order $p$.

\subsection{Bayes Factors}\label{sec.BF}

As mentioned in the introduction the Bayes Factor (BF) of the exact vs the approximated model,
for some fixed data $\by$ is $\frac{P^\alpha_{\bY} (\by)}{P_{\bY} (\by)}$.  Assuming an equal
prior probability for both models, the BF is the posterior odds of one model against the other
(ie. $\frac{p}{1-p}$ with $p=\frac{P^\alpha_{\bY} (\by)}{P_{\bY} (\by) + P_{\bY} (\by)}$ the posterior probability
of the numerical model with  discretization $\alpha$). In terms of model equivalence
an alternative expression conveying the same odds is
$$
\frac{1}{2} \left| 1 -  \frac{P^\alpha_{\bY} (\by)}{P_{\bY} (\by)} \right| .
$$
We will call the above absolute deviance of the BF from 1 the ``ABF''.
In terms of the Jeffreys scale if the ABF is less than 1 (ie. 1 $\leq$ BF $\leq$ 3)
the difference regarding both models is ``not worth more than a bare mention''  \citep{KASS1995, Jeffreys61}.
An even more stringent requirement would be an ABF less than
$\frac{1}{20} = 0.05$, for example, to practically ensure
no difference in both the numerical and the approximated posteriors.
In the next section we see how to bound the \textit{expected} ABF in terms of the 
absolute maximum global error for the numeric FM.

\section{Main result}\label{sec:main}

We now present our main result.

\begin{theorem}\label{theo.main}
For the Expected ABF (EABF) we have that 
\begin{equation}\label{eqn.main}
|| P_{\bY} (\cdot) - P^{\alpha}_{\bY} (\cdot) ||_{TV} =
\int \frac{1}{2} \left| 1 - \frac{P^{\alpha}_{\bY} (\by)}{P_{\bY} (\by)} \right| P_{\bY} (\by) d\by 
\leq  \sqrt{\frac{1}{2 \pi}} \frac{n}{\sigma^*} K |\alpha|^p  \frac{b_i}{n} \sum_{i=1}^n \sum_{j=1}^n | a_{ij} | .
\end{equation}
where $\sigma^* = \left( \int \frac{1}{\sigma} g(\sigma) d\sigma \right)^{-1}$.
\end{theorem}

\begin{proof}
As in \cite{Capistran2016} define the likelihood ratio
$$
R_{\alpha}(\theta) = \frac{P^\alpha_{ \bY | \Phi} (\by | \theta,\sigma)}
{P_{ \bY | \Phi} (\by | \theta,\sigma)} .
$$
Using (\ref{eqn.ll}) and after a simple manipulation we see that
$$
R_{\alpha}(\theta) =
\exp\left[ -\frac{1}{2\sigma^2} \left\{
-2 (\vf_{\theta}^{\alpha} - \vf_{\theta})'\bA(\by - \vf_{\theta}) +
(\vf_{\theta}^{\alpha} - \vf_{\theta})'\bA(\vf_{\theta}^{\alpha} - \vf_{\theta}) \right\}\right] .
$$
Let $\bD_{\theta}^{\alpha} = (\vf_{\theta}^{\alpha} - \vf_{\theta})'$.  Using
(\ref{eqn.global_error}) we have
\begin{equation}\label{eqn.quad_bound}
|\bD_{\theta}^{\alpha} \bA (\bD_{\theta}^{\alpha})'| <
||\bD_{\theta}^{\alpha}||_2^2 ||\bA||_2 = O(|\alpha|^{2p}) 
\end{equation}
where $||\bA||_2$ refers to the induced $L_2$ matrix norm.  Therefore for $|\alpha|$ small
$$
R_{\alpha}(\theta) -1 = \frac{1}{\sigma^2}  \bD_{\theta}^{\alpha} \bA(\by - \vf_{\theta}) + O(|\alpha|^{2p})
$$
since $e^{-x} = 1 - x + O(x^2)$ for $|x|$ small.  Note now that
$$
P^\alpha_{\bY} (\by) = P_{\bY} (\by) +
\int P_{ \bY | \Phi }( \by |  \theta,\sigma) (R_{\alpha}(\theta) -1) P_{\Phi}(\theta,\sigma) d\theta d\sigma .
$$
Dividing by $P_{\bY} (\by)$ we have
\begin{eqnarray}
\left| 1 -  \frac{P^\alpha_{\bY} (\by)}{P_{\bY} (\by)} \right| &=&
\left| \int \left(\frac{1}{\sigma^2}  \bD_{\theta}^{\alpha} \bA(\by - \vf_{\theta}) + O(|\alpha|^{2p})\right)
P_{ \Phi | \bY }( \theta,\sigma | \by ) d\theta d\sigma \right| \nonumber \\
&\leq&  \frac{1}{\sigma^2}  \int \left|  \bD_{\theta}^{\alpha} \bA(\by - \vf_{\theta})\right|
P_{ \Phi | \bY }( \theta,\sigma | \by ) d\theta d\sigma  + |O(|\alpha|^{2p})| .\label{eqn.ABF}
\end{eqnarray}
Now for the EABF we have
$$
\int \left| 1 - \frac{P^{\alpha}_{\bY} (\by)}{P_{\bY} (\by)} \right| P_{\bY} (\by) d\by \leq
\frac{1}{\sigma^2}  \int \int \left|  \bD_{\theta}^{\alpha} \bA(\by - \vf_{\theta})\right|
P_{ \bY | \Phi }( \by |  \theta,\sigma) d\by P_{\Phi}(\theta,\sigma) d\theta d\sigma ,
$$
ignoring the higher order term.
Let $\bu = \bD_{\theta}^{\alpha} \bA$ with $\bu = (u_i)$, then
\begin{equation}\label{eqn.bound1}
 \sigma^{-2} \left|  \bD_{\theta}^{\alpha} \bA(\by - \vf_{\theta})\right| =
 \sigma^{-2} \left| \sum_{i=1}^n u_i (y_i -  f(F_{\theta}(x_i))) \right|
 \leq \frac{1}{\sigma} \sum_{i=1}^n |u_i| b_i \left| \frac{y_i -  f(F_{\theta}(x_i))}{\sigma b_i} \right| .
\end{equation}
Since $\int |x| \frac{1}{\sqrt{2 \pi}} e^{-\frac{x^2}{2}} dx = \sqrt{\frac{2}{\pi}}$ we have
$$
\int \left| 1 - \frac{P^{\alpha}_{\bY} (\by)}{P_{\bY} (\by)} \right| P_{\bY} (\by) d\by \leq
\sqrt{\frac{2}{\pi}} \int \frac{b_i}{\sigma} \sum_{i=1}^n |u_i| P_{\Phi}(\theta,\sigma) d\theta d\sigma ,
$$
and therefore
\begin{equation}\label{eqn.main1}
\int \frac{1}{2} \left| 1 - \frac{P^{\alpha}_{\bY} (\by)}{P_{\bY} (\by)} \right| P_{\bY} (\by) d\by \leq
 \sqrt{\frac{1}{2 \pi}} \frac{b_i}{\sigma^*} \int \sum_{i=1}^n |u_i|  P_\Theta (\theta) d\theta =
 \sqrt{\frac{1}{2 \pi}} \frac{b_i}{\sigma^*} \int ||\bD_{\theta}^{\alpha} \bA||_1 P_\Theta (\theta) d\theta .
\end{equation} 
From (\ref{eqn.global_error}) note that
\begin{equation}\label{eqn.bound2}
|u_i| \leq K |\alpha|^p \sum_{j=1}^n | a_{ij} |
\end{equation}
and from this we obtain the result.
\end{proof}

\begin{corollary}
The EABF tends to zero in the same order as the numerical solver, that is $O(|\alpha|^p)$;
the same happens to the ABF assuming
$$
\int \left| \frac{y_i -  f(F_{\theta}(x_i))}{\sigma b_i} \right|
P_{ \Phi | \bY }( \theta,\sigma | \by ) d\theta d\sigma < \infty.
$$ 
for all $x_i$.
\end{corollary}
\begin{proof}
The result is immediate for the EABF since the bound in (\ref{eqn.main}) is $O(|\alpha|^p)$.
For the ABF use (\ref{eqn.ABF}) and (\ref{eqn.bound1}) to obtain the result.
\end{proof}

\subsection{Remarks on Theorem~\ref{theo.main}}

\begin{itemize}

\item Note that if, as in many inverse problems formulations, the gaussian observation noise variance
 $\sigma^2$ assumed known, then we could think that its prior is a Dirac delta and indeed
$\sigma^* = \left( \int \frac{1}{s} g(s) ds \right)^{-1} = \sigma$.  In fact, we assume in the examples in
sections~\ref{sec:exaODE} and~\ref{sec:exaPDE} that $\sigma$ is known.

\item The bound in Theorem~\ref{theo.main} is a result of the uniform bound for the error
in the numerical solver, as stated in (\ref{eqn.global_error}).  However, a more precise
and at least theoretically interesting bound would be using the integral in the rhs of
(\ref{eqn.main1})
$$
\int ||\bD_{\theta}^{\alpha} \bA||_1 P_\Theta (\theta) d\theta < C .
$$
With this one may obtain an \textit{expected} bound for the EABF.  Since it is an
expected and not an uniform bound, potentially we may allow more solver error
in regions with less a priori mass, etc. 

However, our present approach considers using a
simpler global bound uniform for all $\theta$, perhaps less precise,
but importantly not adding more computational burden on problems already
computationally very demanding, as explained in the next remark. 

\item If we set the absolute uniform global error as
$| f(F^{\alpha}_{\theta}(x_i)) - f(F_{\theta}(x_i)) | \leq K_0 = K |\alpha|^p$ and,
as explained in section \ref{sec.BF},
if we let the $EABF \leq \frac{1}{20} = 0.05$ we expect nearly no difference in the numerical and the
theoretical posterior.  For uncorrelated data we have $a_{ij} = \delta_{i,j}$ and therefore 
$\frac{b_i}{n} \sum_{i=1}^n \sum_{i=1}^n | a_{ij} | = 1$.  Then, setting
$\sqrt{\frac{1}{2\pi}} \frac{n K_0 }{\sigma^*} \leq \frac{1}{20} $ we need
\begin{equation}\label{eqn.main_bound}
K_0  \leq \frac{k}{n} \sigma^* 
\end{equation}
where $k = \frac{1}{20} \sqrt{2 \pi} \approx 0.12$.  That is, we may tolerate an absolute
uniform error of up to 12\% of the expected standard deviation observation noise for $n=1$, 0.12\% for
$n=10$, 0.012\% for $n=100$ etc.  Reasonably, in the light of this analysis of the Bayesian
inverse problems, the tolerated error in the numerical solver is put in terms of the observational
noise expected and the sample size.  When more noise is expected more error may be
tolerated and as the sample size increases the numerical solver needs to be more precise.

This is the main application of our theorem, the uniform global error
to be accepted should be bounded in terms of the observational noise and sample size considered
and not solely as a property of the forward map at hand.   Our bound could be conservative,
but it is intended to be an useful and applicable tool serving as a reference for the numerical
solver precision.

\item Note that the bound in (\ref{eqn.global_error}) may be realized during computation.  That is,
the solver may be applied and its error estimated at the design points $x_i$,
when needing the FM for some specific $\theta$.  If the required bound is exceeded 
then the  discretization $\alpha$ could be altered to comply with the error bound.  That is, 
in practice there is no need to establish   (\ref{eqn.global_error}) theoretically, but rather
by a careful strategy for actual global error estimation (in the numerical analysis literature
these error estimates are derived from the so called output or
\textit{a posteriori} error estimates, but here we use a different name for obvious reasons).
The posterior distribution
in most cases is sampled using MCMC, which requires the approximated likelihood
at each of many iterations; an automatic process of global error estimation and control
will be required in order to comply with  (\ref{eqn.global_error}).

\item The result in the corollary regarding the ABF means that, for any fixed data $\by$ the 
BF  $\frac{P^\alpha_{\bY} (\by)}{P_{\bY} (\by)}$ tends to one in the same order as the solver
order, that is $O(|\alpha|^p)$.  This is a more general version of the main result of
\cite{Capistran2016} using the assumption
$$
\int \left| \frac{y_i -  f(F_{\theta}(x_i))}{\sigma b_i} \right|
P_{ \Phi | \bY }( \theta,\sigma | \by ) d\theta d\sigma < \infty 
$$
(i.e. the posterior expected standardrized absolute residuals are bounded).
In \cite{Capistran2016} they assumed instead that the parameter space $A$ is compact, which
implies the above.  The setting in \cite{Capistran2016} is only for ODE's while ours
considers more general forward maps, including ODEs and PDEs.

\end{itemize}

\section{Example using an ODE: Logistic growth}\label{sec:exaODE}

\begin{figure}
\begin{center}
\includegraphics[height=5.5cm, width=8cm]{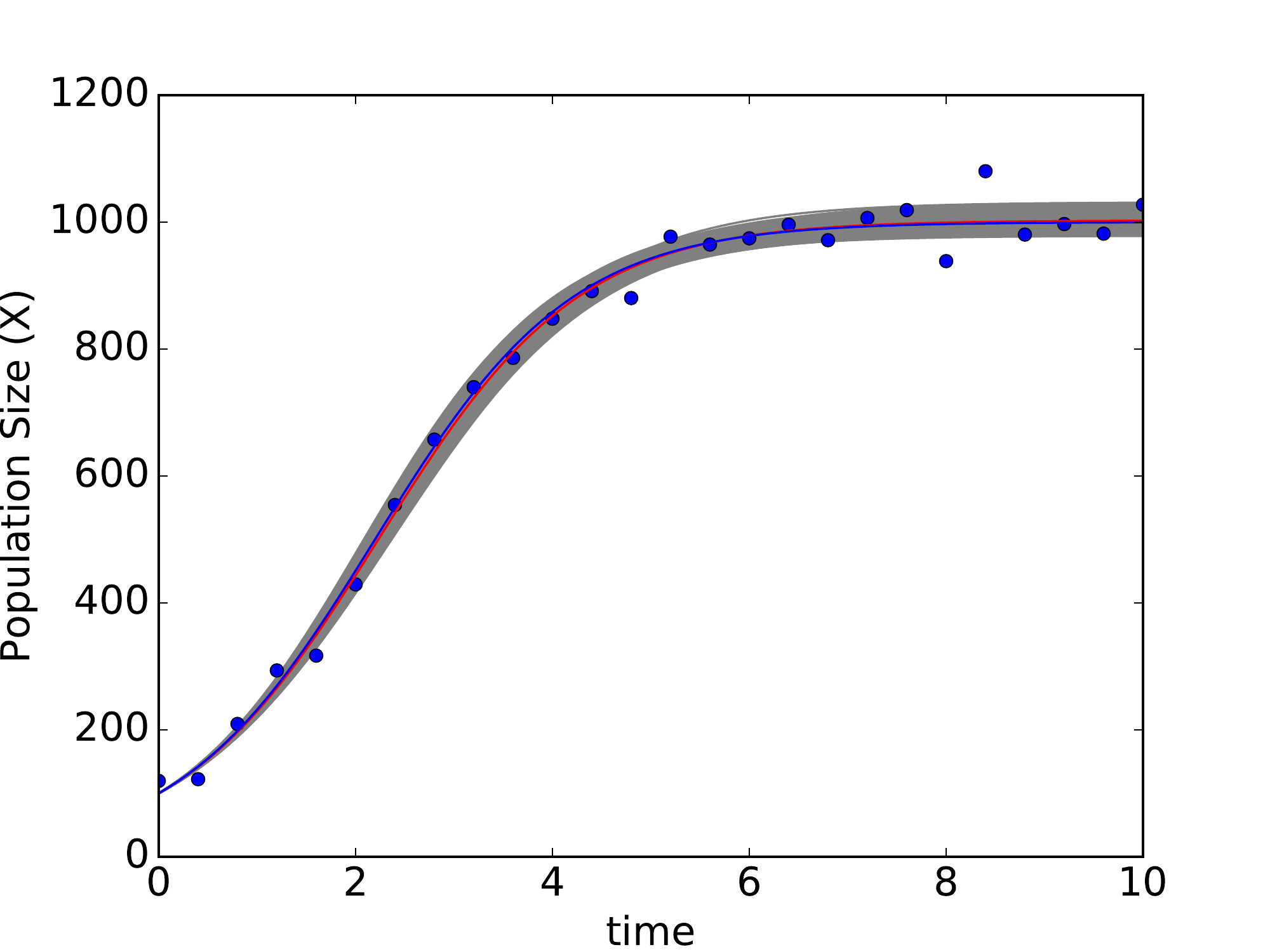}
\caption{\label{fig.LogGr_DataAndFit} Logistic growth example data, true model (blue)
and best (MAP) fit (red).  True parameters are $r=1$, $K=1000$ and $\sigma=30$.
Shaded areas represent the uncertainty in the model fit, as draws from the posterior distribution.}
\end{center}
\end{figure}

As a prove of concept, we base our first numerical study on the logistic growth model which is a common model of
population growth in ecology, medicine, among many other applications
\citep{Forys2003}. 
Let $X(t)$ be, for example, the size of a tumor to time $t$. The logistic growth
dynamics are governed by the following differential equation
\begin{equation}\label{eq:logistic}
 \frac{dX}{dt} = r X(t) (1-X(t)/K), \quad X(0)=X_0
\end{equation}
with $r$ being the growth rate and $K$ the carrying capacity e.g. $\lim_{t\rightarrow \infty} X(t) = K$.  
The ODE in (\ref{eq:logistic}) has an explicit solution equal to
$$
X(t) = \frac{KX_0}{X_0 + (K-X_0)e^{- r t}}.
$$
We simulate a synthetic data set with the error model $ y_i  =X(t_i) + \varepsilon_i$, where $\varepsilon_i \sim \mathcal{N}(0,\sigma^2)$, and  the following parameters $X(0)=100, \quad  r= 1, \quad K=1000,$ $\sigma = 30$.
The data are plotted in Figure \ref{fig.LogGr_DataAndFit}.
We consider $26$ observations at times $t_i$  regularly spaced between $0$ and $10$.
\cite{Capistran2016} also studied this example.

Since we have an analytic solution, if we run a numerical solver on the system we may calculate the
maximum absolute error of the solver, $K_0$ in (\ref{eqn.main_bound}), exactly by comparing with the analytic solution.
Moreover, in Appendix~\ref{sec.appendErrODE} we explain how the
global error may be estimated using Runge-Kutta solver methods.  We run a Cash--Karp RK method,
of order 5 \citep{Cash1990}, which enables us to produce and estimate $\hat{K}_0$ of $K_0$.

The error bound for the FM as stated in (\ref{eqn.main_bound}) is $\frac{k}{26} 30 \approx 0.13$.
The ODE with the initial condition defines our forward map $F_{\theta}$ ($q=p=1$) and we take
$f(x) = x$ as the observational functional.  To sample from the posterior distribution we use
a generic MCMC algorithm called the t-walk \cite{Christen2010}.

Regarding the numerical solver we start with a large step size of $0.1$,
that would maintain numerical stability, and calculate $\hat{K}_0$ and $K_0$.
If the solution does not comply with the bound as calculated by the estimate, that is
$\hat{K}_0 > 0.13$, a new solution is
attempted by reducing the step size by half, until the Runge-Kutta estimated global absolute errors
is within the bound, $\hat{K}_0 \leq 0.13$. 
The results are shown in figure~\ref{fig.LogGr_PostComp}.  Moreover, a fine step size Runge-Kutta
solver was also ran with step size of 0.005, for comparisons.
No difference was observed in both posterior distributions, and the posterior mean and MAP
estimators were basically identical.  However, a 93\% save was obtain by reducing CPU time 
from approximately 100 to 7 min, with 40,000 MCMC iterations.

Since in this case an analytic solution is available we also calculate the exact  maximum
absolute error $K_0$.  The average estimated error was $7.8\cdot10^{-3}$ while the average
exact error was $3.8\cdot10^{-5}$, for the adaptive step size method as describe above.
For the fixed time step method we had estimated and average errors
$2.1\cdot10^{-8}$ and $6.2\cdot10^{-12}$, respectively.  In both cases our error estimates where
some orders of magnitude higher than the true errors.  Certainly the fine (fixed) step solver is
quite more precise reaching more than 5 orders of magnitude larger precision.  However, in the light of theorem
\ref{theo.main} this increase precision is not worth the great increase in CPU time since
indeed results are basically the same using a coarser solver.  And a 93\% CPU time reduction was
obtained nevertheless our error estimates where two orders of magnitude higher.
Numerical errors are now viewed, not in a context free situation, but in the context of this inverse
problem, relative to sample size and expected observational errors, by using the bound in (\ref{eqn.main_bound}).


\begin{figure}
\begin{center}
\begin{tabular}{c c}
\includegraphics[height=4.5cm, width=6.5cm]{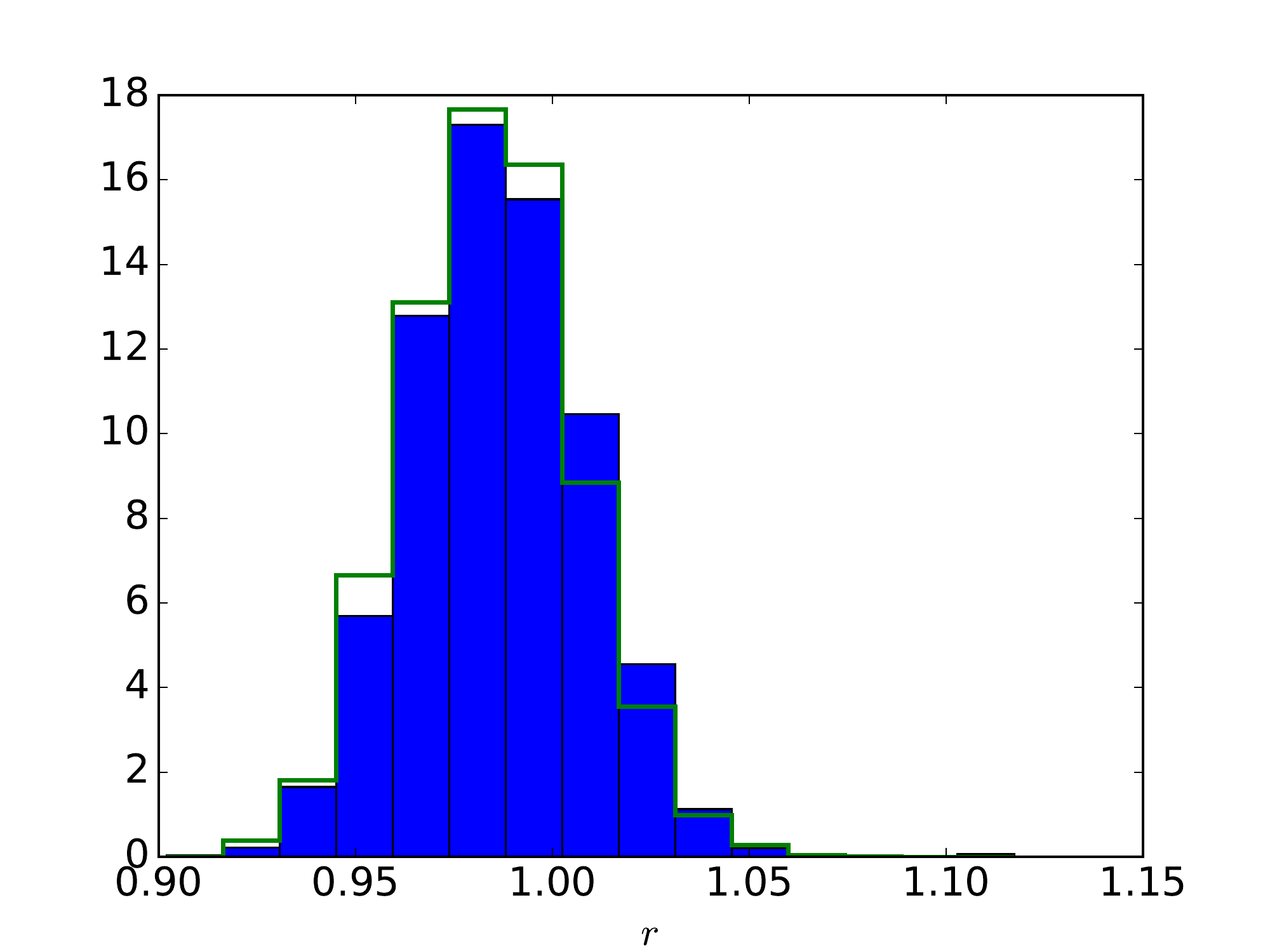} &
\includegraphics[height=4.5cm, width=6.5cm]{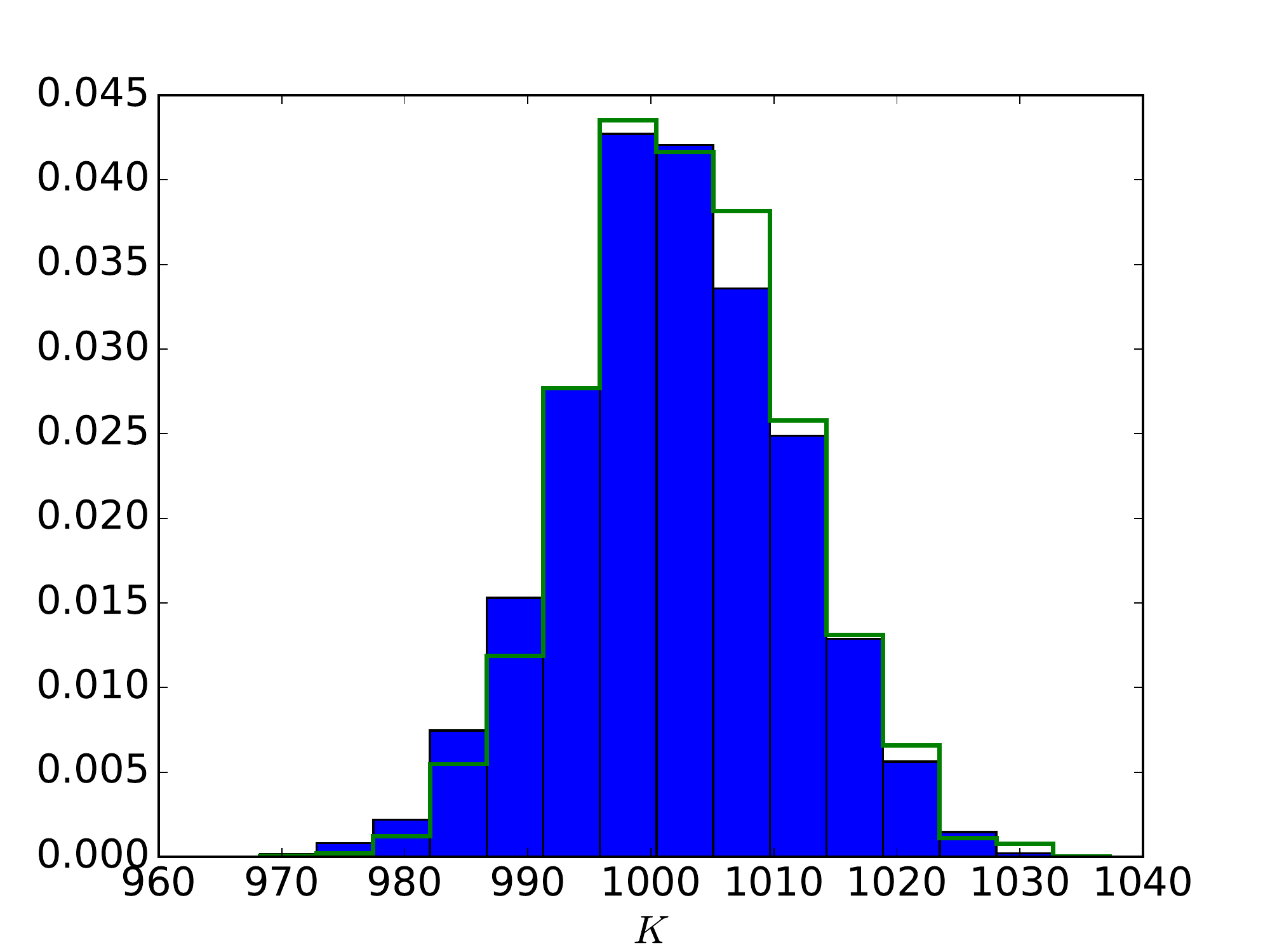} \\
(a) & (b) \\
\end{tabular}
\caption{\label{fig.LogGr_PostComp} Logistic growth example marginal posterior
distributions of $r$ (a) and $K$ (b), $\sigma=30$.
For a fine step size of 0.005 (blue histograms) we obtain basically the same results
(green histograms) as starting from a step size of 0.1 and restricting the estimated maximum
absolute error to the proposed bound, that is $\hat{K}_0 \leq \frac{k}{26} \sigma = 0.13$.
A 92\% CPU time reduction was obtained, from
approximately 103 to 7 min, running 40,000 iterations of the MCMC.}
\end{center}
\end{figure}

\section{Example using a PDE: Burgers' equation}\label{sec:exaPDE}

Burgers' is a fundamental PDE with many applications in fluid mechanics, acoustics, traffic flow modeling, among others
\cite{leveque2002finite}.

Let us consider the Riemann problem for the viscous Burgers' equation
\begin{equation}
\label{eq:burgers}
\begin{split}
u_{t} + uu_{z}&= \epsilon u_{zz}\\
u(z,0)&=
\begin{cases}
u_{L}  & z < z_{0}\\
u_{R} & z_{0} < z,
\end{cases}
\end{split}
\end{equation}
where  $u_{L}$, $u_{R}$, $z_{0}$ and $\epsilon$ are real numbers.
The direct problem is to compute $u = u( z, t)$ solving the initial value problem~(\ref{eq:burgers}).
We assume $u_{L}>u_{R}$ which would correspond to a shock wave in the inviscid limit $\epsilon=0$.

Using the Cole-Hopf transformation we can obtain the solution of problem (\ref{eq:burgers}) in closed form
to obtain
\begin{equation}
\label{eq:soln_burgers}
u( z, t) = u_{L}-\frac{u_{L}-u_{R}}{1 + q(z-z_{0},t) \exp(-\frac{u_{L}-u_{R}}{2\epsilon}(z -z_{0} - ct))}
\end{equation}
where $q( z, t) = \text{erfc}(\frac{ z - u_{R} t}{\sqrt{4\epsilon t}})/\text{erfc}(\frac{ z - u_{L}t}{\sqrt{ 4\epsilon t}})$
and $c=(u_{L}+u_{R})/2$;  see Chapter 4 of Whitham~\cite{Whitham1999}, for details.

The inference problem is to estimate $\theta=(u_{L}-u_{R}, z_{0})$ given observations
$u_{j} = u( z_{1}, t_{j} )$ at a fixed point in space $z_{1}$ for $j=1,..,k$.
That is, the locations are $x_j = ( z_{1}, t_{j} )$.
As in the previous example, the observation operator $f(\cdot)$ is the identity.

To numerically solve the Riemann problem in~(\ref{eq:burgers}), we use a classical 
second-order accurate finite-volume implementation of the viscous Burgers' equation with piecewise 
linear slope reconstruction with outflow boundary conditions, see Leveque~\cite{leveque2002finite}.
In this case, we build a grid with both space and time steps, namely  $\alpha = (\Delta z, \Delta t )$.  The standard procedure
is to fix $\Delta z$ and the temporal grid is determined by a Courant-Friedricks-Levy condition 
\begin{equation}
\label{eq:cfl}
\Delta t_n = c \frac{\Delta z}{\max | u(\cdot, t_{n-1}) |} .
\end{equation}

In the numerical example below we use $c=0.1$.
In this case $|\alpha| = \Delta z$ and $p=2$, that is, the numerical method is order 2, $O(|\alpha|^2)$.
Moreover, error estimates may be obtained from this solution, as explained in Appendix~\ref{sec.appenErrPDE}.

Using the analytic solution in (\ref{eq:soln_burgers}) and the error estimates,
we can establish the exact error $K_0$ and its estimate $\hat{K}_{0}$ as in the previous example.
From $\hat{K}_{0}$ we check if the adaptive method is within the bounds established in (\ref{eqn.main_bound}).
Note that, in the current case $\hat{K}_{0}\approx K_{0}$, provided $K_{0}$ is estimated through interpolation,
as shown in Appendix~\ref{sec.appenErrPDE}.

Again we use synthetic data using  the observation point $z_1=2.0$, $n=6$,
$t=\{0.0,0.1,0.2,\allowbreak 0.3,0.4,0.5\}$ and
$\sigma = 0.0115$ (we choose this standard error to obtain a signal-to-noise ratio of $100$).
The true parameter values are $u_{R}=2$, $u_{L}=1$, $z_{0}=1$.

We estimate the numerical posterior distribution with a high resolution spatial grid of $512$ points
for $z \in [0,4]$, and an adaptive grid, starting with a grid of 128.  If the bound is not met then the grid was doubled to 
256 and thereafter to 512.
We used 20,000 iterations of the twalk MCMC algorithm \cite{Christen2010} to simulate from the
posterior distributions using the high-resolution and adaptive grid solvers.
The execution time for the high resolution and adaptive grids are respectively $10.95$ h
and $4.27$ h respectively.
That is, in this case we obtained a 60\% decrease in CPU, again leading
to basically the same posterior distribution.

\begin{figure}
\begin{center}
\begin{tabular}{c c}
\includegraphics[height=5.cm, width=6cm]{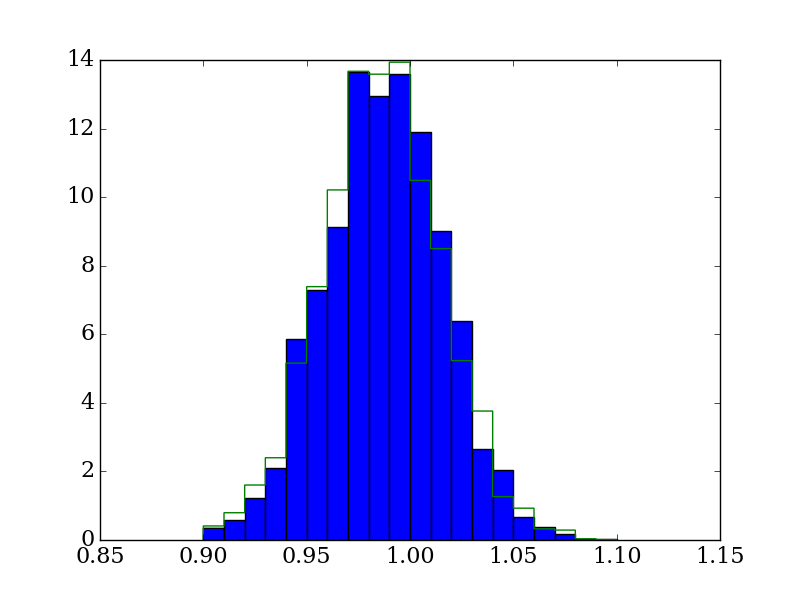} &
\includegraphics[height=5.cm, width=6cm]{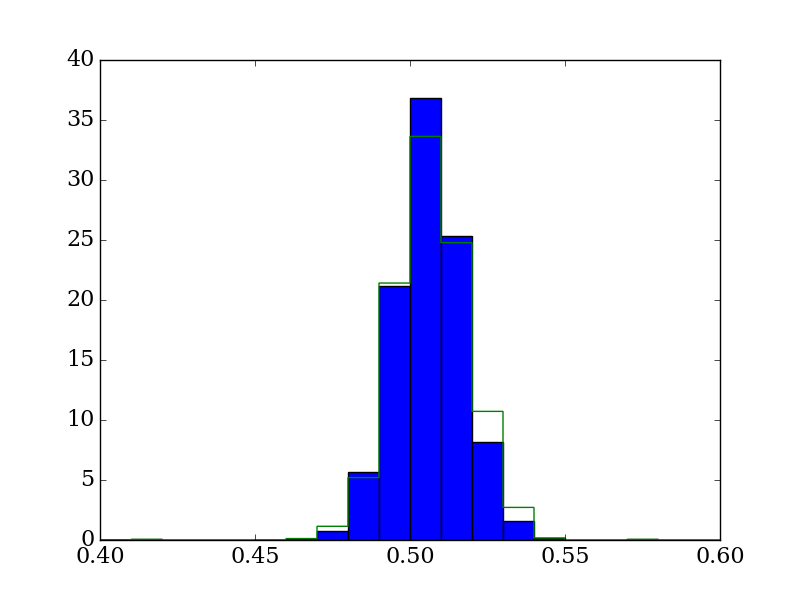} \\
(a) & (b) \\
\end{tabular}
\caption{\small \label{fig.Burgers_PostComp}  Burgers equation example, marginal posterior distribution for
(a) $u_{L}-u_{R}$ and (b) $z_{0}$.  Blue histograms are the result of the adaptive
grid using our bound, transparent green histogram denotes the posterior distribution from the high resolution
grid (see text).  Both histograms are basically the same, with a 60\% decrease in CPU time.}
\end{center}
\end{figure}

\section{Discussion}\label{sec:discussion}

The Bayesian UQ analysis of inverse problems continues to be a very challenging research topic.  In this paper we tried
to contribute to the development of this discipline by analyzing the relationship between error in the numerical solver
of the differential system under study and the corresponding induced error in the numerical posterior distribution.
Once it is established, under regularity conditions, that the numerical posterior distribution tends to the theoretical
posterior as the solver error tends to zero, we still need to decide to which precision run the solver.
Elaborating on previous work \cite{Capistran2016}, theorem~\ref{theo.main} suggest that by carefully
choosing a threshold for the global error in the solver we may
obtain a posterior distribution that is basically error free.  An intuitive perspective on our result is the following:
since there is observational error, we may tolerate certain small amount of error in the solver, which will end up blurred by the
observational error and therefore not noticeable in the posterior distribution.

Indeed, ``small'' is relative to the data error standard deviation $\sigma$ (and to the sample size),
as expressed in (\ref{eqn.main_bound}).  This also suggest that, in Bayesian UQ, solver error could be
viewed in the perspective of the inference problem at hand, potentially allowing for less precise and less computational
demanding solvers that, nevertheless induce basically no error in the resulting posterior distribution and therefore
in the uncertainty quantification of the problem at hand.
 
On the other hand, it is well known that the numerical solver  discretization $\alpha$ must 
be carefully tuned and $|\alpha|$ cannot be
increased to arbitrary values.  In complex cases, beyond some tight limits for the  discretization the numerical solver becomes unstable
and results as in (\ref{eqn:global_error}) cease to be valid.  Indeed, solver order and global error control are properties
valid for ``small'' $|\alpha|$.  Therefore, in a large set of case studies, little room will be available to decide
on an optimal solver  discretization.  Moreover, in very complex PDE solvers, just changing the  discretization is a very
demanding  enterprise, since there is as yet no automatic and reliable way to define the solver grid, as in for
example some spacio-temporal 3D PDE case studies \cite{Cui2011}.

Our examples were proof of concept only, but show promising results.  More importantly, we believe, since the
bound in (\ref{eqn.main_bound}) has a simple form, the posterior adaptive control described in both the
ODE and the PDE examples could be applied in other case studies.
However, this hinges on the availability of reliable, adaptive-like, solver error estimates
which, unfortunately, are not always readily available.  

A natural continuation of this work is to consider the performance of embedded methods such as Cash-Karp RK45,
as well as other feasible alternatives, on \textit{stiff} ODE problems in this Bayesian context.
For PDEs our current interest is on conservation laws
approximating the Forward Map  by the Discontinuous Galerkin method. For this method, error estimates rest on a solid foundation making our approach promising, see \cite{ hesthaven2007nodal, di2011mathematical}.

The study of our results in more complex UQ problems is also left for future research.  Note also that in many PDE
problems the actual unknown is a function, eg. an unknown boundary condition on a scattering problem, needed to
be recovered from data.  In such a case, the theoretical posterior distribution is infinite dimensional and, as we explained
from the onset, we did not consider such case.  The proof of our results in such arbitrary space setting is also left for
future research.

\bibliography{PostErrControl_and_EABF}

\appendix

\section{Global error estimation in Runge-Kutta methods}\label{sec.appendErrODE}

For the reader's convenience, here we briefly describe how we may estimate the global error in solving
the following ODE initial value problem
\begin{equation*}
\frac{X(t)}{dt} = G( t, X, \theta), X(0) = X_0, 
 \end{equation*}
when using a Runge-Kutta (RK) type numerical method, see \cite{Quarteroni2006} chap. 11 for details.
We need the parameters $\theta$
of the ODE system to be fixed and therefore we write $G( t, X) = G( t, X, \theta)$ to ease notation.

A RK method can be written as follows.  Let $\alpha = h > 0$ be a (time) step size and define a uniform grid
such that $t_{n+1} = t_n + h$.  In this case $|\alpha| = h$.  Let 
$
u_{n+1} = u_n + h \sum_{i=1}^s b_i K_i;
$
$u_{n+1}$ is the approximation for $X(t_n)$ and
$
K_i = G(t_n+c_ih,u_n+h\sum_{j=1}^s a_{ij}K_j),\quad i=1,2,\ldots,s .
$
$s$ denotes the number of \emph{stages} of the method.

The components of the vector $\bc' = (c_1,c_2,\dots,c_s)$ need to satisfy
$
c_i=\sum_{j=1}^s a_{ij},\quad i=1,2,\ldots,s 
$.
Let $\bA = (a_{ij})$ and $\bb' = (b_1,b_2,\dots,b_s)$, any RK method is defined by the matrix
$\bA$ and the vectors $\bb$ and $\bc$.
To have  an explicit method we require $a_{ij}=0$ for $j \geq i$, with $i=1,2,\ldots,s$.  We used
an explicit method in our implementation.

The local truncation error $\tau_{n+1}$ at node $t_{n+1}$ of the RK method is defined 
as the error made in step $n+1$ of the solver if starting at the exact value $X(t_n)$, that is
$
\tau_{n+1} = X(t_{n+1}) - X(t_n) - h \sum_{i=1}^s b_i K_i .
$.
The RK method is \emph{consistent}, if $\tau = max_n\vert \tau_n \vert \to 0$ as $h\to 0$. This happens if and only if
$\sum_{i=1}^sb_i=1$.  The RK method is of order $p$ if $\tau = O(h^{p+1})$ as $h\to 0$ and it is known
that $s \geq p$.  The global
(truncation) error at knot $t_n$ is defined as the error made by the solver, that is
$
e_n = X(t_n) - u_n .
$
It is clear that $e_n = \sum_{i=1}^{n} \tau_n$.  Under regularity conditions for a RK method of order $p$ we have
$K_0 = max_n \vert e_n \vert = O(h^p)$ as $h\to 0$.  That is, the maximum absolute global error $K_0$  is of
order $p$ as $h\to 0$.

The strategy described here to estimate $e_n$ is to consider two \textit{embedded} RK methods to solve the system,
one with order $p$ and one
with order $p-1$, both with the same number of stages $s$ and the same matrix $\bA$ and vector
$\bc$, only with different vectors $\bb$ and $\hat{\bb}$, respectively. 

Let $u_{n+1}$ be the $n+1$ estimation of $X(t_{n+1})$ of the $p$ order method and
let $y_{n+1}$ be obtained by the $p-1$ order method by starting at $u_n$, namely
$$
u_{n+1} = u_n + h\sum_{i=1}^s b_i K_i ~~\text{and}~~ y_{n+1}=u_n+h\sum_{i=1}^{s} \hat{b}_i K_i .
$$
An estimation of the local truncation error at $t_{n+1}$ is
$\hat{\tau}_n = u_{n+1} - y_{n+1} = h\sum_{i=1}^{s} (b_i - \hat{b}_i) K_i$
which is basically a byproduct of the $p$ order solver.  

The estimate of the global truncation error at knot $t_n$
is then $\hat{e}_n = \sum_{i=1}^{n} \hat{\tau}_n$.

\section{Global error estimation in the Burgers PDE solver}\label{sec.appenErrPDE}

In order to solve numerically the initial condition problem for the viscous Burgers in (\ref{eq:burgers}) 
we used a second order explicit finite-volume method to handle the advective flux. 
On the other hand, second order time-stepping is accomplished through Crank-Nicolson 
updating implicitly in the viscosity term. We discretize the solution using homogeneous 
Neumann conditions on a space interval $I=[0,4]$ adding two ghost cells at each endpoint. 
In order to march in time we enforce the Courant-Friedrichs-Levy condition through the 
time-stepping rule (\ref{eq:cfl})
where $c=0.1$ and $u_{h}$ is the numerical solution at time $t$. Discretization is implemented 
setting a number of space points, e.g. $N=2^{9}$. Of note, the time step is adapted to obtain
the solution at prescribed observation times $t=\{0.0,0.1,0.2,0.3,0.4,0.5\}$.

If we denote the numerical solution by $u_{h}$, then the residue is 
\begin{equation}
\label{eq:residue}
R_{h}(u)=
\frac{\partial u_{h}}{\partial t} + u \frac{\partial u_{h}}{\partial z} - \epsilon \frac{\partial^{2}u_{h}}{\partial z^{2}} .
\end{equation}

We apply our main result using the following after the fact
error estimate for continuous approximations to nonlinear viscous hyperbolic conservation laws;
see theorem 5.2 of Cockburn \cite{cockburn1999simple}.
Let $v$ be the entropy solution (of problem~(\ref{eq:burgers})) and let $u$ be a 
continuos approximation. Then
\begin{equation}
\label{eq:traffic_bound}
||u(T)-v(T)||_{L^{1}(\mathbb{R})}\leq\Phi(v_{0},u,T)
\end{equation}
where
\begin{displaymath}
\Phi(v_{0},u,T)=||u(0)-v_{0}||_{L^{1}(\mathbb{R})}+||R_{h}(u)||_{L^{1}(0,T)\times\mathbb{R}}
+C(u)\sqrt{\epsilon}
\end{displaymath}
and
\begin{displaymath}
C^{2}(u) = 8|u|_{L^{\infty}(0,T;TV(R))}|u|_{L^{1}(0,T;TV(R))} .
\end{displaymath}

Note that the finite volume method that we have used lets the residual go to zero quadratically. 
Hence, we have the following application of Cockburn theorem. Let $v(t)$ denote the analytic solution 
(\ref{eq:soln_burgers}) of problem (\ref{eq:burgers}),
and let $u_{h}$ denote the finite volume 
solution at times
$t^{n}$, $n=1,...,N$.  Following Cockburn analysis, let us denote by $\mathcal{U}$ the set of all the interpolates $u$ such
that $u(t^{n})=u_{h}(t^{n})$ for $n=1,...,N$. Then application of Theorem 5.2 of Cockburn
gives
\begin{displaymath}
||u_{h}(t^{n})-v(t^{n})||_{L^{1}(\mathbb{R})}\leq\inf_{u\in\mathcal{U}}\Phi(v_{0},u;t^{n}),\;n=1,..,N
\end{displaymath}

Let us define
\begin{equation}
\label{eq:ratio}
r(u_{h},t^{n})=\frac{\Phi(v_{0},u,T)}{||u_{h}(t^{n})-v(t^{n})||_{L^{1}(\mathbb{R})}}.
\end{equation}
In order to estimate $\hat{K}_{0}\approx K_{0}$ we take $2^{N}+1$ grid points on the spatial domain 
$z\in[0,4]$ for $N=6,7,8,9,19$ and interpolate $r(u_{h},1)=1+K_{0}h^{2}$, where $h=1/\Delta z$.

\end{document}